\pdfoutput=1

\documentclass[journal]{IEEEtran}

\usepackage{amsmath}
\usepackage{graphicx}
\usepackage{amsthm}
\usepackage{amssymb}
\usepackage{pgfplots}
\usepackage{nicefrac}

\def\R{{\mathbb R}}
\def\N{{\mathbb N}}
\def\x{{\mathbf x}}

\def\ie{{\it i.e., }}

\let\oldnl\nl
\newcommand{\nonl}{\renewcommand{\nl}{\let\nl\oldnl}}

\DeclareMathOperator*{\argmin}{arg\,min}
\DeclareMathOperator*{\minimize}{\text{minimize}}

\newlength
\figureheight
\newlength
\figurewidth

\graphicspath{{figures/}}

\begin{document}

\title{Variations on the CSC model}

\author{Ives Rey-Otero, Jeremias Sulam, Michael Elad
\thanks{All authors are with the Computer Science Department, Technion - Israel Institute of Technology}}

\maketitle

\begin{abstract}

    Over the past decade, the celebrated sparse representation model has
    achieved impressive results in various signal and image processing tasks.
    A convolutional version of this model, termed convolutional sparse coding
    (CSC), has been recently reintroduced and extensively studied.
    CSC brings a natural remedy to the limitation of typical sparse enforcing
    approaches of handling global and high-dimensional signals by local,
    patch-based, processing.
    While the classic field of sparse representations has been able to cater
    for the diverse challenges of different signal processing tasks by
    considering a wide range of problem formulations, almost all available
    algorithms that deploy the CSC model consider the same $\ell_1 - \ell_2$
    problem form.
    As we argue in this paper, this CSC pursuit formulation is also too
    restrictive as it fails to explicitly exploit some local characteristics of
    the signal.
    This work expands the range of formulations for the CSC model by
    proposing two convex alternatives that merge global norms with local
    penalties and constraints.
    The main contribution of this work is the derivation of efficient and
    provably converging algorithms to solve these new sparse coding formulations.

\end{abstract}


\section{Introduction}\label{sec:intro}

The sparse representation model~\cite{bookSparse} is a central tool for a wide
range of inverse problems in image processing, such as
denoising~\cite{elad2006image,mairal2009non},
super-resolution~\cite{romano2014single,yang2010image},
image deblurring~\cite{yu2012solving,dong2013nonlocally} and more.
This model assumes that natural signals can be represented as a sparse linear
combination of a few columns, called atoms, taken from a matrix called
dictionary.
The problem of recovering the sparse decomposition of a given signal over a
(typically overcomplete) dictionary is called \emph{sparse coding} or pursuit. Such an
inverse problem is usually formulated as an optimization objective seeking to
minimize the $\ell_0$ pseudo-norm, or its convex relaxation, the $\ell_1$-norm,
while allowing for a \emph{good}\footnote{The desired representation accuracy,
or fitting, is problem dependent and it varies for different applications.}
signal reconstruction.
An effective deployment of the sparse representation model calls for the
identification of a dictionary that suites the data treated.
This is known as the \emph{dictionary learning} problem, of finding the best
sparsifying dictionary that fits a large set of signal
examples~\cite{aharon2006rm,engan1999method}.

Alas, when it comes to the need to process global high-dimensional signals
({\it e.g.,} complete images), the sparse representation model hits strong barriers.
Dictionary learning is completely intractable in such cases due to its too high
memory and computational requirements. In addition, the global pursuit fails to
grasp local varying behaviors in the signal, thus leading to inferior treatment
of the overall data.
Because of these reasons, it has become a common practice to split the global
signal into small overlapping blocks, or patches, identify the dictionary that
best models these patches, and then sparse code and reconstruct each of these
blocks independently before averaging them back into a global
signal~\cite{elad2006image}.
Although practical and effective~\cite{mairal2008sparse}, this patch-based
strategy is inherently limited since it does not account for the natural
dependencies that exist between adjacent or overlapping patches, and therefore it
cannot ensure a coherent reconstruction of the global signal
\cite{sulam2015expected,papyan2016multi}.

This limitation of the patch-based strategy has been tackled in two ways.
One way maintains the patch-based strategy while extending it by modifying the
objective so as to bridge the gap between local prior and global
reconstruction.
This is achieved either by taking into account the self-similarities of natural
images~\cite{dong2013nonlocally,mairal2009non}, by exploiting their multi-scale
nature~\cite{papyan2016multi,mairal2008learning,sulam2014image}, or by
explicitly requiring the reconstructed global signal to be consistent with the
local prior~\cite{zoran2011learning,sulam2015expected}.
The second way consists in dropping the heuristic patch-based strategy altogether in
favor of global, yet computationally tractable and locally-aware, models. Such is the case of the CSC
\cite{grosse2012shift,thiagarajan2008shift,rusu2014explicit}, allowing
the pursuit to be performed directly on the global signal by imposing a
specific banded convolutional structure on the global dictionary.
This implies, naturally, that the signal of interest is a
superposition of a few local atoms shifted to different positions.
And so, while the CSC is a global model, it has patch-based flavor to it and in addition, learning its dictionary is within reach~\cite{bristow2013fast}.

Recent years have seen a renewed interest in the CSC model, including a
thorough theoretical analysis along with new pursuit and dictionary learning
algorithms for it, and its deployment to problems such as image inpainting, super-resolution, dynamic range
imaging, and pattern classification~
\cite{bristow2013fast,heide2015fast,kong2014fast,wohlberg2014efficient,gu2015convolutional,yellin2017blood,serrano2016convolutional}.
Nevertheless, the research activity on the CSC model is still in its infancy.
In particular, while the classic sparse representation model has assembled an
extensive toolbox of problem formulations, diverse sparsity promoting penalty
functions along with countless pursuit algorithms (with greedy, relaxation and
Bayesian alternatives), most pursuit approaches to recover the CSC
representation $\Gamma$ from a global signal $X$ and a convolutional dictionary
$D$ rely on minimizing the same $\ell_2 - \ell_1$ objective, namely
\begin{equation}
    \minimize_\Gamma \frac1{2} \|X - D \Gamma \|^2_2 + \lambda \|\Gamma\|_1,
    \label{}
\end{equation}
where $\lambda$ is a Lagrangian parameter.
As we show in this work, this problem formulation is too restrictive and dull.
Indeed, both terms in this formulation, the $\ell_2$ reconstruction term and
the $\ell_1$ sparsity promoting penalty, are global quantities - as is the
scalar Lagrangian parameter $\lambda$ that controls the trade-off between them.
This contrasts with state-of-the-art patch-based methods where sparsity is
controlled locally, typically through a per-patch constraint on the maximum
number of non-zeros or on the \emph{maximal allowed patch
error}~\cite{elad2006image}.
While one would hope for the CSC pursuit to optimally scatter non-zero
coefficients in a way that best serves the signal, we unfortunately observe that
this not to be the case in practice.
Instead, solutions to the above problem typically exhibit sparsity patterns that
have little relation with the signal local complexity.
This calls for alternative problem formulations where local sparsity and
local representation errors are explicitly taken into account in the global
model.
An additional motivation for an alternative formulation of the CSC pursuit stems from the findings of
\cite{papyan2017working}, which is the first work to derive a theoretical
analysis framework for the CSC model.
In order to leverage the convolutional structure in this pursuit problem, the
authors in~\cite{papyan2017working} advocate for a new notion of local
sparsity.
In particular, they provide recovery and stability guarantees conditioned
on the sparsity of each representation portion responsible for encoding
individual patches, as opposed to the traditional global $\ell_0$ norm.
The CSC pursuit formulations proposed in the present work aim at explicitly
controlling the sparsity level in these portions of the representation vectors,
called stripes.
The first formulation employs the $\ell_{1,\infty}$ norm as the sparsity promoting
function, providing a convex relaxation of the $\ell_{0,\infty}$ pseudo-norm
that was introduced in~\cite{papyan2017working} and explored further in~\cite{plaut2018mpcsc,wohlberg2017convolutional}.
The second formulation controls the sparsity of the stripes by
considering the maximum reconstruction error on each patch simultaneously, via
an $\ell_{2,\infty}$ norm. Such an approach is motivated by 
patch averaging techniques that have been successfully deployed for
denoising and other inverse problems~\cite{elad2006image,mairal2008sparse}.
We derive, for each of these two formulations, simple, efficient, and provably converging algorithms.

The remainder of the paper is organized as follows.
Section~\ref{sec:conv:sparse:coding} reviews common notations and definitions
for the CSC model.
Section~\ref{sec:need:structured} examines the behavior of the classic
$\ell_2 - \ell_1$, in particular its tendency to overuse simple atoms and
encode the signal by aggregation of these coarse atoms along with a spatial
instability of the global representation.
We then propose two alternate formulations, the $\ell_2 - \ell_{1,\infty}$ and
$\ell_{2,\infty} - \ell_1$ in
Section~\ref{sec:l1inf} and Section~\ref{sec:l2inf}, respectively.
Both sections focus on the derivation of algorithms to solve the
respective formulations along with experiments to illustrate their behavior and
performance.
Section~\ref{sec:conclusion} contains a final discussion.

\section{Convolutional Sparse Coding}
\label{sec:conv:sparse:coding}

\begin{figure}[!t]
	\centering
    \includegraphics[trim = 250 50 150 50, width=0.6\textwidth]{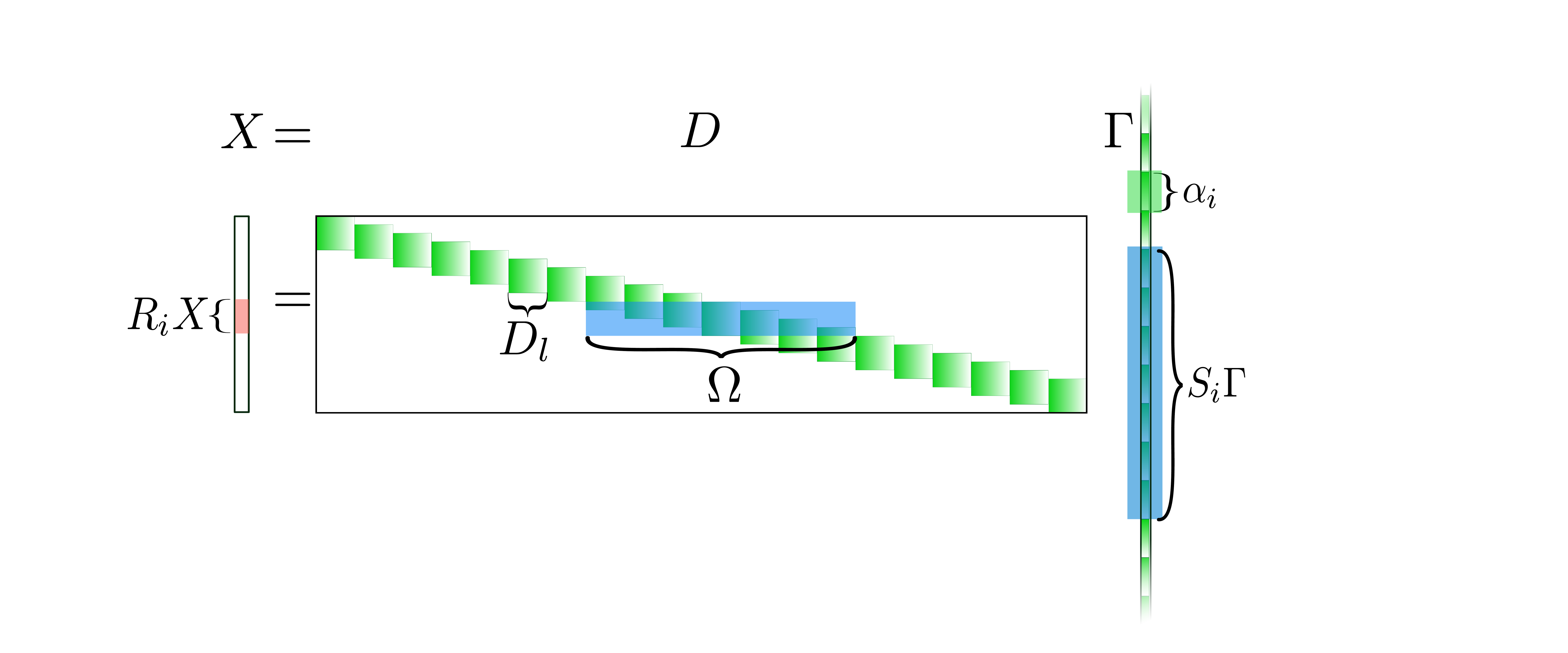}
    \caption{
        Illustration of the CSC model for the $1$D case.
        At the global scale, the image $X$ can be decomposed into the product
        of the global convolutional dictionary $D$ and a global sparse
        representation $\Gamma.$
        At the patch scale, the patch $R_i X$ can be decomposed into the
        product of the stripe dictionary $\Omega$ and the stripe representation vector $S_i\Gamma$.
    }
    \label{fig:CSC:notations}
\end{figure}

This work uses the terminology first introduced in~\cite{papyan2017working}.
The CSC model assumes that an image can be decomposed as $X =
D\Gamma$.
An image of size $H\times W$ is represented in its vectorized form as a
vector $X$ of length $N=HW$ and the corresponding global convolutional
dictionary $D$ is of size $ N \times Nm$.
$D$ is built as the concatenation of $m$ (block-) circulant matrices of size
$N\times N$, each representing one convolution.
These convolutions employ small
support filters of size $n\times n$, thus causing the above-mentioned circulant
matrices to be narrowly banded.
Another way to describe D is by combining  all the shifted versions of a local
dictionary $D_l \in  \R^{n^2 \times m}$ composed of the $m$ vectorized 2D filters.
Such construction is best illustrated by expressing the global signal in
terms of the local dictionary, $X = \sum_{i=1}^N R_i^T D_l \alpha_i$, where
$R_i^T$ is the operator that positions the patch $D_l\alpha_i$ in the $i^{th}$ location
and pads the rest of the entries with zeros.
The quantity $D_l\alpha_i$ is called a slice, with $\alpha_i$ being the portion
of the sparse representation vector $\Gamma$, called needle, that encodes the slice.
It is important to stress that slices are not patches but rather simpler components
that are combined to form patches.

To better understand which parts of the dictionary $D$ and of the sparse vector
$\Gamma$ represent an isolated patch, it is convenient to consider the patch
extraction operator $R_i$ and apply it to the system of equations $X = D \Gamma$.
This yields the system $R_i X = R_i D \Gamma$ consisting of the $n^2$ rows
relating to the patch pixels.
Due to the banded structure of $D$, the extracted rows $R_i D$ contain only a
subset of $(2n-1)^2 m$ columns that are not trivially zeros.
Denoting by $S_i^T$ the operator that extracts such columns and rewriting our
system of equations as $R_i X = R_i D S_i^T S_i \Gamma$ make two interesting
entities come to light.
The first is the vector $S_i\Gamma$, a subset of $(2n-1)^2 m$
coefficients of $\Gamma$ called the \emph{stripe} that entirely encodes the patch
$R_iX$.
The second entity is the sub-matrix $\Omega = R_i D S_i^T \in \R^{n^2 \times
(2n-1)^2 m}$, called the \emph{stripe dictionary}, which multiplies the stripe vector
$S_i\Gamma$ to reconstruct the patch.
Figure~\ref{fig:CSC:notations} summarizes these definitions and notations,
employed in the remainder of the paper.

\section{The need for structured sparsity}
\label{sec:need:structured}

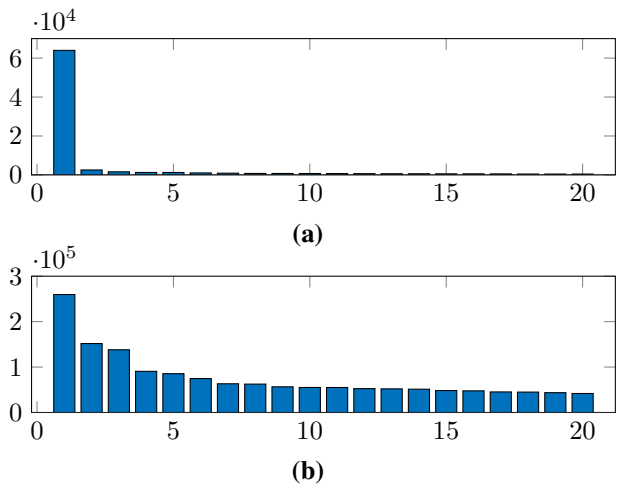
\begin{figure}[!t]
	\centering
	\setlength\figurewidth{0.45\textwidth}
	\setlength\figureheight{0.10\textwidth}
%
%
\definecolor{mycolor1}{rgb}{0.00000,0.44700,0.74100}%
\begin{tikzpicture}

\begin{axis}[%
width=0.951\figurewidth,
height=\figureheight,
at={(0\figurewidth,0\figureheight)},
scale only axis,
bar shift auto,
xmin=-0.2,
xmax=21.2,
ymin=0,
ymax=70000,
axis background/.style={fill=white}
]
\addplot[ybar, bar width=8, fill=mycolor1, draw=black, area legend] table[row sep=crcr] {%
1	63972\\
2	2549\\
3	1596\\
4	1278\\
5	1262\\
6	1027\\
7	900\\
8	767\\
9	756\\
10	734\\
11	712\\
12	651\\
13	616\\
14	600\\
15	564\\
16	563\\
17	528\\
18	465\\
19	464\\
20	453\\
};
\addplot[forget plot, color=white!15!black] table[row sep=crcr] {%
-0.2	0\\
21.2	0\\
};
\end{axis}
\end{tikzpicture}%

    {\bf (a)}

%
%
\definecolor{mycolor1}{rgb}{0.00000,0.44700,0.74100}%
\begin{tikzpicture}

\begin{axis}[%
width=0.951\figurewidth,
height=\figureheight,
at={(0\figurewidth,0\figureheight)},
scale only axis,
bar shift auto,
xmin=-0.2,
xmax=21.2,
ymin=0,
ymax=300000,
axis background/.style={fill=white}
]
\addplot[ybar, bar width=8, fill=mycolor1, draw=black, area legend] table[row sep=crcr] {%
1	259668\\
2	151740\\
3	138030\\
4	90610\\
5	85463\\
6	74653\\
7	63192\\
8	62499\\
9	56539\\
10	55051\\
11	55043\\
12	52543\\
13	52026\\
14	51323\\
15	48329\\
16	47668\\
17	45282\\
18	44937\\
19	43617\\
20	41820\\
};
\addplot[forget plot, color=white!15!black] table[row sep=crcr] {%
-0.2	0\\
21.2	0\\
};
\end{axis}
\end{tikzpicture}%

    {\bf (b)}
    \caption{
        Representation test for the image \texttt{barbara}.
        Number of non-zero coefficients for each of the 20 most
        commonly used atoms in each representation.
        {\bf (a)} CSC $\ell_1 - \ell_2$ formulation
        {\bf (b)} sparse coding every patch individually via OMP.
        The sparsity level in both representations is set so as to reach a
        reconstruction error of $30$dB.
        In both cases, the most popular atom is the DC component.
        The $\ell_1 - \ell_2$ CSC formulation leads to one atom being used
        predominantly.
        Since each patch is reconstructed independently, the patch-based
        reconstruction leads to a denser global representation and uses diverse
        local atoms.
    }
    \label{fig:histo}
\end{figure}

To illustrate the problematic behavior of the CSC model in its most common
formulation, we examine the structure obtained in the global
representation vector in the following experiment.
A natural image $X$ is sparse coded using the $\ell_1 - \ell_2$
formulation\footnote{The $\ell_1 - \ell_2$ minimization is carried out
using the slice-based algorithm proposed in~\cite{papyan2017convolutional}.},
resulting in the decomposition $D\Gamma_\text{CSC}$.
The local dictionary considered, $D_l$, is the DCT dictionary
($n=8$, $m=64$) and $\lambda$ is set so as
to reach a reconstruction error of $30$~dB.
For comparison, fully-overlapping patches of the same image are sparse coded
individually using the Orthogonal Matching Pursuit (OMP) with the same local
dictionary.
The error threshold in the OMP is set so as to achieve a representation error
of $30$~dB after patch-averaging, so as to match the CSC experiment.
This results in a set of needles, one per each patch, and these can be
concatenated into a global sparse representation $\Gamma_\text{OMP}$ that has the same length and structure as $\Gamma_\text{CSC}$.

Figures~\ref{fig:histo} {\bf (a)} and {\bf (b)}  depict how often the first 20 atoms in the
local dictionary are used in $\Gamma_\text{CSC}$ and $\Gamma_\text{OMP}$ respectively.
In the CSC representation vector, one atom is predominantly used, namely the DC
atom.
In fact, most of the needles in $\Gamma_\text{CSC}$  contain at most only one
active atom, and many of them (about $70\%$) remain completely empty.
Note that while the OMP algorithm in the patch-based approach encodes the patches using the local dictionary
atoms alone, the CSC pursuit encodes the entire image using the atoms as well as
their shifts.
This allows fewer local atoms to be used in the CSC representation.
Indeed, the system of equations $X = D\Gamma$ admits an infinite number of
solutions even with a local dictionary $D_l$ containing as few as two atoms,
which would be, on the other hand, insufficient to reliably reconstruct individual
patches.
What these plots show in fact is that, in the CSC model, the juxtaposition of
the simplest atoms shifted at different locations accounts for most of its
expressiveness.
This tendency leads to a series of problems.
On one hand, this is of importance for any dictionary learning algorithm that builds on
the $\ell_1 - \ell_2$ formulation, since the predominant use of one filter
prevents most atoms from being properly updated and learned.

Another tendency of the CSC model we would like to expose is the fact  that the
global representation obtained is spatially unstable.
By putting too much emphasis on the spatial arrangement of slices, the sparse
representation at one point of the image is overly affected by a distant
structure in the image -- as happens, for example, with the image borders.
To demonstrate this, let us consider two image crops shifted from one another by a few
pixels. The impact of the border location on the spatial distribution of
non-zero coefficients is illustrated in
Figure~\ref{fig:struct:spatial:instability}, which shows the difference between
the respective sparsity maps (with proper compensation of the shift).
Oddly, the global distribution of atoms in the image is globally affected most
noticeably in smooth regions of the image.

\begin{figure}[!t]
	\centering

    \begin{minipage}[t]{.23\textwidth}
        \begin{centering}
            
            \includegraphics[width=\textwidth]{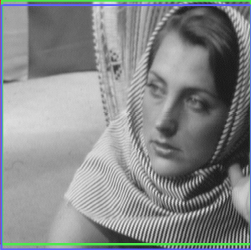}
        \end{centering}
    \end{minipage}
    \begin{minipage}[t]{.23\textwidth}
        \begin{centering}
            
            \includegraphics[width=\textwidth]{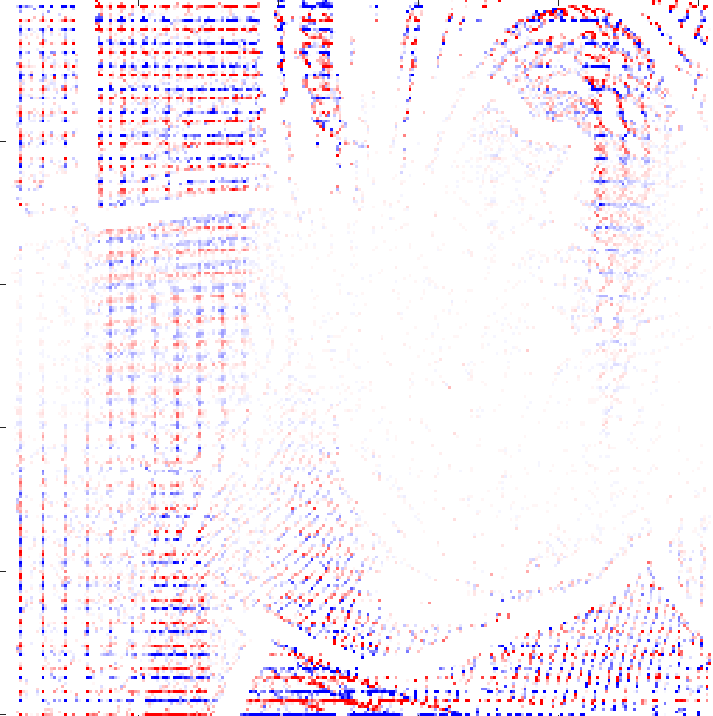}

            \setlength\figurewidth{1.1\textwidth}
	        \setlength\figureheight{0.01\textwidth}
%
%
\begin{tikzpicture}

\begin{axis}[%
width=0.768\figurewidth,
height=\figureheight,
at={(0\figurewidth,0\figureheight)},
scale only axis,
point meta min=-1.5,
point meta max=1.5,
axis on top,
xmin=0.5,
xmax=258.5,
y dir=reverse,
ymin=0.5,
ymax=252.5,
axis line style={draw=none},
ticks=none,
colormap={mymap}{[1pt] rgb(0pt)=(0,0,1); rgb(31pt)=(1,1,1); rgb(32pt)=(1,1,1); rgb(63pt)=(1,0,0)},
colorbar horizontal
]
\end{axis}
\end{tikzpicture}%

        \end{centering}
    \end{minipage}
    \caption{
        Spatial instability of the CSC sparse representation.
        {\bf (a)} We consider two crops setups (in green and blue) only differing
        by a vertical shift of 4 pixels.
        We sparse code each sub-image with the convolutional dictionary based
        on the DCT.
        {\bf (b)}
        Sparsity difference: each pixel in the figure corresponds to the
        difference between the $\ell_1$-norm of the needles in each
        representation (after shifting back the representation by 4 pixels so
        as to consider the needles that are in both representations).
        Note that on smooth regions, representation is spatially sparse and
        follows a grid pattern.
        The $\ell_1 - \ell_2$ formulation leads to juxtaposition of slices,
        whose spatial arrangement is sensitive to the smallest distant
        variations of the signal.
    }
    \label{fig:struct:spatial:instability}
\end{figure}

Note that it is common in practice to deploy the CSC model not directly on the
image itself but rather after applying a local mean subtraction and contrast normalization of the
signal.
This has the effect of mitigating, to some extent, the spatial instability to large distance
interaction by breaking the connections between distant structure.
However, this does not prevent the inherent tendency of the $\ell_1 - \ell_2$
global formulation to use too few atoms and compensating for this by aggregating
overlapping shifts.
We will see that by anchoring the CSC pursuit locally, as in the proposed
alternate formulations, it is possible to get hold of such tendency.

\section{The $\ell_2 - \ell_{1,\infty}$ CSC formulation}
\label{sec:l1inf}

The first alternate formulation that we explore drops the global $\ell_1$ as a
sparsity promoting penalty and uses instead a mixed norm function, adding an
explicit and local control of sparsity.
This is motivated by the work in~\cite{papyan2017working}, whose analysis
centers around a new notion of local sparsity, the $\ell_{0,\infty}$. This
measure, instead of quantifying the total number of non-zeros in a vector,
reports the $\ell_0$ norm of the \emph{densest} stripe: 
\begin{equation}
\|\Gamma\|_{0,\infty} = \max_i \|S_i \Gamma \|_0.
\end{equation}
Such a localized norm is a somewhat more appropriate measure of sparsity in the
convolutional setting, since with it one is able to significantly improve on
the theoretical guarantees for the CSC model \cite{papyan2017working}.
Although that work established that the $\ell_2 - \ell_1$ formulation
approximates the solution to an $\ell_{0,\infty}$ problem, it also conjectured
that further improvement could be achieved by considering a new  $\ell_{1,\infty}$-norm.
This norm, defined as $
\|\Gamma\|_{1,\infty} = \max_i \|S_i \Gamma \|_1$, will be the center of our
current discussion:  the $\ell_2 - \ell_{1,\infty}$ formulation,
\begin{equation}
    \min_\Gamma \frac1{2} \|X - D \Gamma \|^2_2 + \lambda \|\Gamma\|_{1,\infty}.
    \label{eq:csc:l1inf}
\end{equation}
The $\ell_{1,\infty}$ is nothing but a mixed norm on the global representation $\Gamma$.
Mixed-norms have been commonly used in signal processing to promote various types of
structure in the sparsity pattern~\cite{kowalski2009sparse}.
In the context of the CSC model, using this mixed norm is expected to promote a
distribution of non-zero coefficients that makes use of more diverse local atoms
and is less affected by the global attributes of the image.

This formulation, in fact, first appeared in the work
of~\cite{wohlberg2017convolutional}, which proposed a global ADMM formulation
to iteratively minimize the loss in Equation~\eqref{eq:csc:l1inf}.
Unfortunately, one of the steps in their proposed iterative process requires of
yet another ADMM solver, resulting in a generally inefficient algorithm.
The complexity of this approach is aggravated by the need of a multi-block
ADMM, which requires careful parameter tuning and does not enjoy the
convergence properties of the standard ADMM.

\subsection{The proposed algorithm}

Recalling the $\ell_2 - \ell_{1,\infty}$ formulation in
Equation~\eqref{eq:csc:l1inf}, consider $N$ splitting variables
$\{\gamma_i\}_{i=1}^N$, so as to rewrite the problem  equivalently as
\begin{equation}
    \begin{aligned}
        & \minimize_{\substack{\Gamma, \{ \gamma_i \} }}
        & &   \frac1{2} \| Y - D\Gamma \|_2^2 + \lambda \max_i \| \gamma_i \|_1 \\
        & \text{ subject to }
        & & \forall i, \; \gamma_i = S_i\Gamma.
    \end{aligned}
    \label{eq:csc:l1inf:ADMM:standard:form}
\end{equation}
This constrained minimization problem is handled by considering its augmented Lagrangian:
\begin{equation}
    \begin{aligned}
        & \minimize_{\substack{\Gamma, \{ \gamma_i \},   \{ u_i \},}}
        & &   \frac1{2} \| Y - D\Gamma \|_2^2 + \lambda \max_i \| \gamma_i \|_1 \\ 
         &&& + \frac{\rho}{2} \sum_i \| \gamma_i - S_i \Gamma + u_i \|^2_2,
    \end{aligned}
    \label{eq:csc:l1inf:ADMM:augmented:Lagragian}
\end{equation}
where $ \{ u_i \}^N_{i=1}$ denote the scaled dual-variables associated with
each equality constraint $\gamma_i = S_i\Gamma$.
The ADMM algorithm~\cite{boyd2011distributed} minimizes this augmented
Lagrangian by alternatively updating the variable $\Gamma$ and the
set of splitting variables $\{ \gamma_i \}_{i=1}^N$.
Formally, an iteration of the ADMM algorithm consists of the following steps:
\begin{align}
    %
    %
    \Gamma^{(k)} :=&  \argmin_{\Gamma}
    \begin{aligned}[t]
        &\frac1{2} \| Y - D\Gamma \|_2^2 \label{admm:1} \\
      +& \frac{\rho}{2} \sum_i \| \gamma_i^{(k-1)} - S_i \Gamma + u_i^{(k-1)}\|_2^2.
    \end{aligned}\\
    %
    %
    \{\gamma_i^{(k)}\}  :=& \argmin_{ \{\gamma_i\} }
    \begin{aligned}[t]
        & \lambda \max_i \|\gamma_i\|_1  \label{admm:2}  \\
        +&  \frac{\rho}{2} \sum_i \| \gamma_i - S_i \Gamma^{(k)} + u_i^{(k-1)}\|_2^2.
    \end{aligned}\\
    %
    %
    u_i^{(k)} :=& \quad u_i^{(k-1)} + \gamma_i^{(k)} - S_i \Gamma^{(k)}.    \label{admm:3} \\ \nonumber
\label{}
\end{align}
The update of $\Gamma$ in Equation~\eqref{admm:1} is straightforward, as it is a
least-square minimization that boils down to solving the linear system of equations
\begin{align}
    \left( D^T D + \rho \sum_i S_i^T S_i \right) \Gamma =
    \begin{aligned}[t]
        & D^T Y \\ %
        &+ \rho \sum_i S_i^T(\gamma_i + u_i).
    \end{aligned}
    \label{eq:update:gamma:least:square}
\end{align}
Bearing in mind that fast (possibly GPU) implementations are available for the convolution
$D^T$ and the transpose convolution $D$, and using the fact that $\sum_i S_i^T
S_i  = (2n-1)^2 I$, this \emph{regularized} least-square minimization can be carried
out efficiently and reliably via a few iterations of the conjugate gradient
method~\cite{kelley1999iterative}.

The updates of the variables $\{ \gamma_i \}^N_{i=1}$ in
Equation~\eqref{admm:2} are seemingly more complicated, due to the max
operation between the different stripes and the fact that they overlap.
To make it more manageable, we cast the Problem~\eqref{admm:2} in
epigraph form as
\begin{equation}
    \begin{aligned}
        & \minimize_{\substack{ \{ \gamma_i \}, t}}
        & & \lambda t  + \frac{\rho}{2} \sum_i \| \gamma_i - S_i \Gamma^{(k+1)} + u_i^{(k)}\|_2^2,\\
        & \text{subject to}
        & & \forall i, \quad  \|\gamma_i\|_1 \leq t.
    \end{aligned}
    \label{eq:admm:2:epigraph}
\end{equation}
Here, the initial problem with variables $\{\gamma_i \}_{i=1}^N $ has just been
replaced with an equivalent minimization over variables $\{\gamma_i \}_{i=1}^N$
and $t$.
Note that, for a fixed value of variable $t$, this new objective
in Equation~\eqref{eq:admm:2:epigraph} is now separable in the variables $ \{ \gamma_i
\}_{i=1}^N$. 
More precisely, it can be broken down into $N$ separate minimization problems
\begin{equation}
    \begin{aligned}
        \bar \gamma_i (t) := & \argmin_{\substack{\gamma_i}}
        & &  \| \gamma_i - S_i \Gamma^{(k)} + u_i^{(k-1)}\|_2^2,\\
        & \text{subject to}
        & &  \|\gamma_i\|_1 \leq t.
    \end{aligned}
    \label{eq:admm:2:t:fixed}
\end{equation}
Each of these is simply a projection onto the
$\ell_1$-ball~\cite{duchi2008efficient} that can be performed via the shrinkage operator\footnote{
$\mathcal{S}_\lambda(\x)$ denotes the shrinkage operator, formally
%
    $ \mathcal{S}_\lambda(\x) = \text{sign}(\x) \odot \text{max} \left( |\x| - \lambda, 0 \right),$
with $\odot$ denoting the element-wise product.}:
\begin{equation}
    \bar \gamma_i(t) = \mathcal{S}_{\lambda^\ast} \left( S_i \Gamma^{(k)} - u_i^{(k-1)}  \right),
    \label{}
\end{equation}
where the shrinkage parameter $\lambda^\ast$ can be efficiently estimated by
sorting the vector's coefficients and computing over them a cumulative sum (see
\cite{duchi2008efficient} for details).

In this way, solving the initial problem~\eqref{admm:2} boils down to finding
the optimal $t$ leading to the minimum of the objective, namely
$\{\gamma_i^{(k)} \}_{i=1}^N = \{\gamma_i(t^\ast) \}_{i=1}^N$ with
\begin{equation}
    t^\ast := \argmin_{t} \left( \lambda t + \sum_i \| \bar \gamma_i(t) - S_i \Gamma^{(k)}  + u_i^{(k-1)} \|^2_2 \right).
    \label{eq:sum:shortest:distance}
\end{equation}
As a sum of an affine function and squared  distances to the $\ell_1$ ball of
radius $t$, the previous objective is a convex function of $t$.
Indeed, the distance to the $\ell_1$ ball is a convex function of the radius $t$
(see Proposition~\ref{prop:conv} in Appendix~\ref{sec:appendix}).
and it can therefore be minimized efficiently via a simple binary-search.

This simple algorithm, by not involving an over-sensitive Lagrange multiplier
setting, and by enjoying the convergence properties of the standard ADMM
compares favorably with the method described
in~\cite{wohlberg2017convolutional}.

\subsection{Experiments}

\begin{figure}[!t]
	\centering
	\setlength\figurewidth{0.14\textwidth}
	\setlength\figureheight{0.14\textwidth}

    \begin{minipage}{.22\textwidth}
        \begin{centering}
        \includegraphics[width=0.8\textwidth]{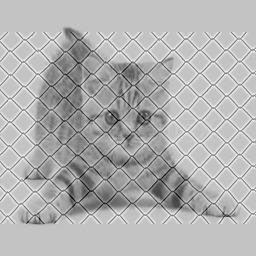}

            \texttt{cat}

        \end{centering}
    \end{minipage}
    \begin{minipage}{.22\textwidth}
        \begin{centering}
        \includegraphics[width=0.8\textwidth]{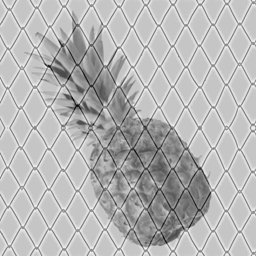}

            \texttt{pineapple}

        \end{centering}
    \end{minipage}

    \vspace{1em}

    \begin{minipage}{.22\textwidth}
        \begin{centering}
        \includegraphics[width=0.8\textwidth]{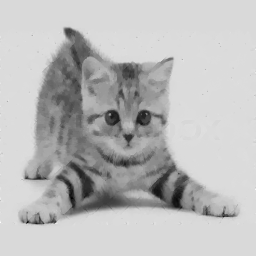}

        \includegraphics[width=0.8\textwidth]{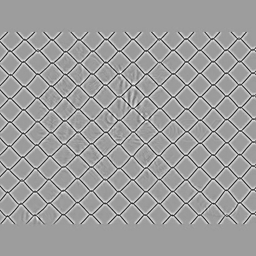}

            $\ell_2 - \ell_{1}$

        \end{centering}
    \end{minipage}
    \begin{minipage}{.22\textwidth}
        \begin{centering}
        \includegraphics[width=0.8\textwidth]{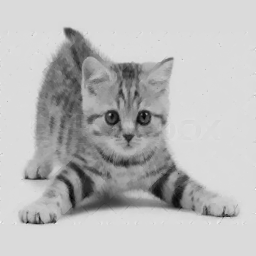}

        \includegraphics[width=0.8\textwidth]{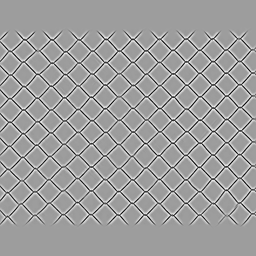}

            $\ell_2 - \ell_{1,\infty}$

        \end{centering}
    \end{minipage}

    \vspace{1em}

    \begin{minipage}{.22\textwidth}
        \begin{centering}

        \includegraphics[width=0.8\textwidth]{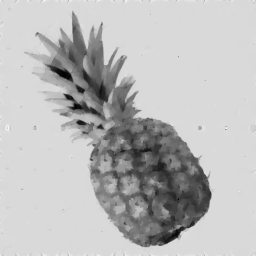}

        \includegraphics[width=0.8\textwidth]{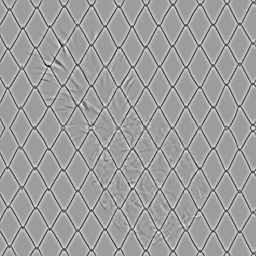}

            $\ell_2 - \ell_{1}$

        \end{centering}
    \end{minipage}
    \begin{minipage}{.22\textwidth}
        \begin{centering}

        \includegraphics[width=0.8\textwidth]{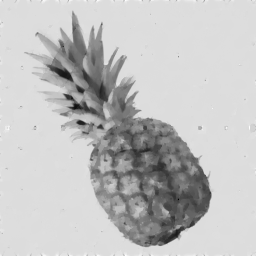}

        \includegraphics[width=0.8\textwidth]{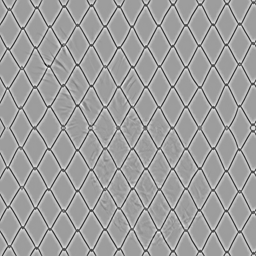}

            $\ell_2 - \ell_{1,\infty}$

        \end{centering}
    \end{minipage}
    \caption{
        Noiseless texture-cartoon separation.
        Comparing the $\ell_2 - \ell_{1,\infty}$ and $\ell_2 - \ell_{1}$ formulations.
        The input images consist of the test image \texttt{cat} and
        \texttt{pineapple}.
    }
    \label{fig:cartoon:texture}
\end{figure}

We illustrate the $\ell_2 - \ell_{1,\infty}$ formulation on the texture-cartoon
separation task.
This problem consists in decomposing an input image $X$ into a piecewise smooth
component (cartoon) $X_c$ and a texture component $X_t$ such that $X = X_c + X_t$.
The typical prior for the cartoon component $X_c$ is based on the total variation
norm, denoted $\|X_c \|_\text{TV}$, which penalizes oscillations.
In addition, we propose to assume that the texture component $X_t$ admits a
decomposition $X_t = D_t\Gamma$ where $D_t$ is a convolutional texture
dictionary and $\Gamma$ is the solution of the $\ell_2 - \ell_{1,\infty}$ CSC
formulation.
Under these assumptions, the task of texture and cartoon separation boils
down to a minimization problem over three variables: the cartoon component
$X_c$, the CSC representation $\Gamma$ and a convolutional texture dictionary
$D_t$, namely
\begin{equation}
    \minimize_{\Gamma, D_t, X_c}  \frac{1}2 \| X - D_t\Gamma - X_c \|^2_2  + \lambda \|\Gamma \|_{1,\infty} + \zeta \|X_c\|_\text{TV},
    \label{eq:csc:l1inf:cartoon:texture}
\end{equation}
with parameter $\zeta$ controling the level of TV regularization penalizing oscillations in $X_c$.
Such minimization is carried out iteratively in a block-coordinated manner
until convergence. Each iteration consists of the three following steps:
\begin{align}
    %
    %
    X_c^{(k+1)} :=&  \argmin_{X_c}
    \begin{aligned}[t]
        &\frac1{2} \| X - D_t^{(k)} \Gamma^{(k)} -X_c \|_2^2  \label{eq:cartoon:1}  \\
        &+ \zeta \| X_c  \|_\text{TV}
    \end{aligned}\\
    %
    %
    \Gamma^{(k+1)}  :=& \argmin_{\Gamma}
    \begin{aligned}[t]
        &\frac1{2} \| X - D_t^{(k)} \Gamma -X_c^{(k+1)} \|_2^2  \label{eq:cartoon:2} \\
        & + \lambda \|\Gamma \|_{1,\infty}
    \end{aligned}\\
    %
    %
    D_t^{(k+1)}  :=& \argmin_{D_t} \frac1{2} \| X - D_t \Gamma^{(k+1)} -X_c^{(k+1)} \|_2^2. \label{eq:cartoon:3} 
\end{align}
A TV denoiser\footnote{The TV denoiser used here is the publicly available
implementation of~\cite{chan2011augmented}.} is used to solve
Problem~\eqref{eq:cartoon:1} while Problem~\eqref{eq:cartoon:2} relies on our
$\ell_2 - \ell_{1,\infty}$ solver.
For the dictionary update, one option is to use a standard patch-based
dictionary learning such as K-SVD using overlapping patches as
training sets and the needles of the current $\Gamma$ estimate.
However this would not be consistent with the CSC model. Indeed, the patch
would then be assumed to stem from the local dictionary alone, disregarding all the
contributions of shifted atoms to its reconstruction.
We adopt instead a more coherent alternative that was recently proposed
in~\cite{plaut2018mpcsc} in which standard dictionary update procedures are
adapted to a convolutional setting and carried out via conjugate gradient
descent~\cite{kelley1999iterative} in conjunction with fast convolution computations.
The proposed method is applied to the test images \texttt{cat} and
\texttt{pineapple}, the results of our method are shown in Figure~\ref{fig:cartoon:texture} along with the results
from the $\ell_1 - \ell_2$ based method in~\cite{papyan2017convolutional}

\section{The $\ell_{2,\infty} - \ell_1$ CSC formulation}
\label{sec:l2inf}

We move on to consider our second formulation, of explicitly incorporating a local
control on the CSC model. This is inspired by the patch-based strategy for image denoising and other inverse problems.
Recall that patch-based sparse denoising
methods~\cite{elad2006image,mairal2008sparse} control the sparsity level on each patch by
upper-bounding the patch reconstruction error. We will borrow such an idea, and translate it into the convolutional setting.

For a noisy image $Y$, patch methods rely on a global objective of the form
\begin{equation}
    \begin{aligned}
        &\minimize_{\substack{ \{\alpha_i\}, X}}
             & &  \frac{\lambda}2 \|X -Y\|_2^2 + \sum_i \|\beta_i \|_0 \\
        & \text{ subject to }
             & & \forall i, \, \| D_l \beta_i - R_i X \|_2^2 \leq T,
    \end{aligned}
    \label{eq:pa:denoising}
\end{equation}
where $\beta_i$ is the sparse vector for the patch $R_i X$ and the upper-bound
$T$ over the patch reconstruction error is typically set to
$Cn^2\sigma_\text{noise}^2$, the assumed patch noise level (up to a
multiplicative constant).
This is typically solved via a block-coordinate descent algorithm, which means first initializing $X = Y$ and seeking the sparsest
$\alpha_i$ for each patch via the set of local problems
\begin{equation}
    \begin{aligned}
        & \minimize_{\substack{ \beta_i}}
        & &  \|\beta_i \|_0 \\
        & \text{ subject to } & & \| D_l \beta_i - R_i Y \|_2^2 \leq T, \\
    \end{aligned}
    \label{eq:pa:l2inf}
\end{equation}
which yields a reconstruction for each overlapping patch and, in turn, an intermediary global reconstruction $\frac1{n^2} \sum_i R_i^T D_L \beta_i$.
While state-of-the-art methods typically consider approximate solutions through
greedy pursuit algorithms, it is also possible to consider an $\ell_1$
relaxation of the same sparse coding problem. We will employ the latter option
in order to benefit from the resulting convexity of the problem.

The second stage of the block-coordinate descent algorithm consists in updating
the estimate of $X$, the restored image, by solving the least-square problem in
closed form~\cite{elad2006image} according to:
\begin{equation}
    X = \left(\lambda I + \sum R_i^T R_i \right)^{-1} \left( \lambda Y + \sum_i R_i^T D_L \beta_i \right),
    \label{eq:pa:averaging}
\end{equation}
essentially averaging the input signal $Y$ with the patch-averaging
estimate $ \frac1{n^2} \sum_i R_i^T D_L \beta_i$.

In order to bring this classic approach into a convolutional setting,
note that the CSC global representation $\Gamma$ can be decomposed
into its constituent \emph{needles}, and so $\sum_i \|\alpha_i\|_1 =
\|\Gamma\|_1 $.
Recalling the definitions and notations in
Section~\ref{sec:conv:sparse:coding}, a patch from the reconstructed image $R_i
X$ in the CSC model can be equivalently written as $R_i X = R_i D\Gamma = \Omega S_i \Gamma$.
With these elements, the problem in \eqref{eq:pa:denoising} can be naturally
transformed into
\begin{equation}
    \begin{aligned}
        & \minimize_{\substack{ \{\alpha_i\}, X}}
        & &  \frac{\lambda}2 \|X -Y\|_2^2 + \| \Gamma\|_1  \\
        & \text{ subject to } & &  \forall i, \;  \|\Omega S_i \Gamma - R_i X \|_2^2 \leq T. \\
    \end{aligned}
    \label{eq:csc:l2inf:pa}
\end{equation}
One might indeed adopt a similar block-coordinate descent strategy for this
problem as well. After an initialization of $X=Y$, the first step considers the
resulting $ \ell_{2, \infty} - \ell_1 $ formulation:
\begin{equation}
    \begin{aligned}
        & \minimize_{\substack{ \Gamma}}
        & &  \|\Gamma \|_1 \\
        & \text{ subject to } & &  \forall i, \;  \|\Omega S_i \Gamma - R_i Y \|_2^2 \leq T, \\
    \end{aligned}
    \label{eq:csc:l2inf}
\end{equation}
where the constraint on patch reconstruction considers the stripe dictionary.
Again, the second stage consists in updating the estimate of $X$ by solving the least-square problem
\begin{equation}
    X = \left(\lambda I + \sum_i R_i^T R_i \right)^{-1} \left( \lambda Y + \sum_i R_i^T \Omega S_i \Gamma \right). 
    \label{eq:pa:averaging:CSC}
\end{equation}
whose solution, since $\sum_i R_i^T \Omega S_i \Gamma = n^2 D \Gamma $ and since $ \sum_i R_i^T R_i = n^2 I $, boils
down to an average between the input image and the intermediary global
reconstruction $D\Gamma$.
In this manner, and similarly to the patch-averaging strategy, the trade-off
between sparsity and reconstruction is controlled locally via an upper-bound on
the reconstruction error of each individual patch.
However, while in the original method each vector $\beta_i$ encodes one patch
in disregard with other patches, now each needle $\alpha_i$ becomes part of
various stripes $S_i\Gamma$ and therefore contributes in various patches.
In other words, the classic patch-averaging approach performs these pursuit
independently, whereas this convolutional counterpart will need to update all
needles jointly.

In what follows, we show that this seemingly complex problem can in fact be
addressed by using traditional $\ell_1$ solvers such as  the Fast Iterative
Shrinkage-Tresholding Algorithm (FISTA)~\cite{beck2009fast} in conjunction with
the Parallel Proximal Algorithm (PPXA).

\subsection{Proposed algorithm}

PPXA is a generic convex optimization algorithm introduced by Combettes and
Pesquet~\cite{combettes2008proximal,combettes2011proximal} that extends the Douglas-Rachford
algorithm and aims to minimize an objective of the form
\begin{equation}
    \minimize_x \sum_i^N f_i(x),
    \label{eq:ppxa:generic:form}
\end{equation}
where each $f_i$ is a convex function that admits an easy-to-compute proximal
operator~\cite{parikh2014proximal,bauschke2011convex}.
Recall that the proximity operator $\text{prox}_{f_i}(y): \R^N \rightarrow
\R^N$ of $f_i$ is defined by
\begin{equation}
    \text{prox}_{f_i}(y) := \argmin_x f_i(x) + \nicefrac{1}{2} \| x - y\|^2_2.
    \label{eq:def:proximal:operator}
\end{equation}
In our context, PPXA offers a way to manage the explicit use of
overlapping stripes.
Indeed, by encapsulating each inequality constraint into its corresponding
indicator function, the objective in Equation~\eqref{eq:csc:l2inf} can be recast as
a sum, namely
\begin{equation}
    \begin{aligned}
        \minimize_\Gamma \sum_{i=1}^N \left(  \frac1{N} \|\Gamma \|_1 + \mathcal{I}_{ \{ \|\Omega S_i \Gamma - R_i Y \|^2_2 \leq T \} } \right),
    \end{aligned}
    \label{eq:csc:l2inf:sum:fi}
\end{equation}
where $ \mathcal{I}_{  \{ \| \Omega S_i \Gamma - R_i Y \|^2_2 \leq T \}} $
denotes the indicator function\footnote{The indicator function  $\mathcal{I}_S$
equals $0$ inside the set $S$ and $\infty$ elsewhere.} on the constraint
feasibility set.
The successful deployment of the PPXA algorithm for this problem depends on our ability
to compute, for each patch, the proximal operator
\begin{equation}
    \begin{aligned}
        \text{prox}_{f_i}(\Gamma) :=  \argmin_{\hat \Gamma} & \quad \|\hat \Gamma \|_1
                                 + \frac{1}{2 N \mu} \| \Gamma - \hat\Gamma \|_2^2 \\
                                  & + \mathcal{I}_{  \{ \|\Omega S_i \hat \Gamma - R_i Y \|^2_2 \leq T \} },
    \end{aligned}
    \label{eq:csc:l2inf:proximal:ind}
\end{equation}
with parameter $\mu$ scaling the least-square term.
%
%
The solution to the above problem is also the solution to a Lagrangian
\begin{equation}
    \argmin_{\hat \Gamma} \|\hat \Gamma \|_1 + \frac{1}{2 N \mu} \| \Gamma - \hat \Gamma \|_2^2 + \lambda_i^\ast \| R_i ( D \hat \Gamma - Y ) \|^2_2,
    \label{eq:csc:l2inf:proximal:lagrangian}
\end{equation}
in which the Lagrange multiplier is set to an optimal value $\lambda_i^\ast$: the \emph{smallest} Lagrange multiplier such
that the inequality constraint is satisfied.
Observe that, while transitioning from
Equation~\eqref{eq:csc:l2inf:proximal:ind} to
Equation~\eqref{eq:csc:l2inf:proximal:lagrangian}, we moved from $\Omega$ to
$D$, in order to pose the algorithm w.r.t. the global dictionary.
%
Fortunately, for a given Lagrangian multiplier $\lambda_i$, such objective can
be efficiently minimized by a proximal gradient method such as
(ISTA)~\cite{daubechies2004iterative} or its fast version
FISTA~\cite{beck2009fast}.
Indeed, denoting $g_i(\hat \Gamma, \lambda_i) :=  \frac1{2 N \mu} \|\Gamma - \hat \Gamma \|^2_2 + \lambda_i \| R_i( D \hat \Gamma - Y ) \|_2^2$,  ISTA and FISTA revolve around the update step
\begin{equation}
    \begin{aligned}
        \hat \Gamma^{(k+1)} =  \mathcal{S}_{t_k} \Big(
                   \hat \Gamma^{(k)} + t_k \frac{ \partial g_i}{\partial \hat \Gamma }(\hat \Gamma^{(k)}, \lambda_i ) \Big),\\
    \end{aligned}
    \label{csc:l2inf:FISTA:step}
\end{equation}
where $t_k$ denotes the step-size\footnote{For convergence, the step-size
$t_k$ must satisfy $t_k \leq \frac{1}{\lambda_\text{max}}$, where $\lambda_\text{max}$ denotes the maximum eigenvalue of $\nabla g_i$ which can
be approximated efficiently via the power method.}.
The dominant effort here is the evaluation of the gradient of $g_i$ with
respect to $\hat \Gamma$. This boils down to the computation of convolutions,
for which fast GPU implementations are available.
Running FISTA successively with warm-start initialization allows to estimate
the minimizer for different values of $\lambda_i$ with only few extra
iterations. This allows to use a binary-search scheme to estimate the
optimal Lagrange multiplier $\lambda_i^\ast$ which in turn provides the
solution to the proximal operator in Equation~\eqref{eq:csc:l2inf:proximal:ind}.

Armed with this procedure to compute the proximal operators, an iteration of
the PPXA algorithm boils down to the following steps:
\begin{enumerate}
    \item Compute the proximal operators for each patch
    \begin{equation}
        \begin{aligned}
            \forall i=1\dots N, \quad \hat \Gamma_{i}^{(l)} = \text{prox}_{f_i} (\Gamma_{i}^{(l)}),
        \end{aligned}
        \label{xxx}
    \end{equation}
    following the procedure described above. The evaluations can be carried out in parallel.
    \item Aggregate the solutions
        \begin{equation}
            \begin{aligned}
            \hat \Gamma^{(l)} = \frac{1}{N} \sum_i^N \hat \Gamma_{i}^{(l)}.
            \end{aligned}
            \label{}
        \end{equation}
    \item Update the estimate of $\Gamma$ along with the auxiliary variables $\Gamma_{i}$
            \begin{align}
                \forall i, \quad  \Gamma_{i}^{(l+1)} & =  \Gamma_{i}^{(n)} + \rho_l \left(2 \hat \Gamma^{(l)} - \Gamma^{(l)} - \hat \Gamma_{i}^{(l)} \right), \\
                \Gamma^{(l+1)}  & = \Gamma^{(l)} + \rho_l (\hat\Gamma^{(l)}  - \Gamma^{(l)}), \nonumber
            \label{}
            \end{align}
\end{enumerate}
where $\rho_l$ denotes the relaxation parameter~\footnote{To guaranty
convergence, the relaxation parameters $(\rho_l)$ must satisfy $\sum_{l\in\N}
\rho_l (2-\rho_l) = +\infty$.} on this iteration.
The sequence of sparse vector estimates $\Gamma^{(l)}$ is proven to
converge to the solution of the $\ell_{2,\infty} - \ell_1$ CSC
problem~\eqref{eq:csc:l2inf} \cite{combettes2008proximal}.
Note that using FISTA in conjunction with PPXA makes it possible to take full
advantage of GPU hardware and high-level libraries for fast convolutions, in
contrast with most sparse coding algorithm that operate in the Fourier domain
\cite{heide2015fast,wohlberg2014efficient}.

\subsection{Extension via weighted stripe dictionary}

\begin{figure}[!t]
	\centering
	\setlength\figurewidth{0.35\textwidth}
	\setlength\figureheight{0.20\textwidth}
%
%
\begin{tikzpicture}

\begin{axis}[%
width=0.951\figurewidth,
height=\figureheight,
at={(0\figurewidth,0\figureheight)},
scale only axis,
xmin=0,
xmax=1,
xlabel style={font=\color{white!15!black}},
xlabel={$\theta$},
ymin=0,
ymax=0.7,
axis background/.style={fill=white},
legend style={legend cell align=left, align=left, fill=none, draw=none}
]
\addplot [color=red]
  table[row sep=crcr]{%
0	0.671233661025983\\
0.05	0.632318796853118\\
0.1	0.574778852254317\\
0.15	0.552136716457857\\
0.2	0.476962287645994\\
0.25	0.430920230733513\\
0.3	0.397228108790076\\
0.35	0.358161844503145\\
0.4	0.325306095578929\\
0.45	0.300669153933545\\
0.5	0.262105198287655\\
0.55	0.238881783152957\\
0.6	0.222447473224087\\
0.65	0.191131801937804\\
0.7	0.174467713052683\\
0.75	0.15717904283284\\
0.8	0.124896186422896\\
0.85	0.115590653048501\\
0.9	0.103812421544325\\
0.95	0.0984530839237949\\
1	0.0937893371221044\\
};
\addlegendentry{$\|n^2 \bar D_l S_i \Gamma - R_i Y \|_2$}

\addplot [color=blue]
  table[row sep=crcr]{%
0	0.0971030693014666\\
0.05	0.0990625571573739\\
0.1	0.0889645363780638\\
0.15	0.0975342193455485\\
0.2	0.0989042923109837\\
0.25	0.0935245291801624\\
0.3	0.0999372712709953\\
0.35	0.0933836558244023\\
0.4	0.0873531729824051\\
0.45	0.0826342933270665\\
0.5	0.0904868148678619\\
0.55	0.0949670048412508\\
0.6	0.0929583746062494\\
0.65	0.0932807708004747\\
0.7	0.0880010240411081\\
0.75	0.0914880582737192\\
0.8	0.0817810686968104\\
0.85	0.0861873473464832\\
0.9	0.0855348467916863\\
0.95	0.0931388303491864\\
1	0.0937893371221044\\
};
\addlegendentry{$\| \Omega_\theta S_i \Gamma - R_i Y \|_2$}

\addplot [color=green, dashed]
  table[row sep=crcr]{%
0	0.1\\
0.05	0.1\\
0.1	0.1\\
0.15	0.1\\
0.2	0.1\\
0.25	0.1\\
0.3	0.1\\
0.35	0.1\\
0.4	0.1\\
0.45	0.1\\
0.5	0.1\\
0.55	0.1\\
0.6	0.1\\
0.65	0.1\\
0.7	0.1\\
0.75	0.1\\
0.8	0.1\\
0.85	0.1\\
0.9	0.1\\
0.95	0.1\\
1	0.1\\
};
\addlegendentry{$T$}

\end{axis}
\end{tikzpicture}%

    {\bf (a)}
    \vspace{1em}

	\setlength\figurewidth{0.18\textwidth}
	\setlength\figureheight{0.18\textwidth}
    \begin{minipage}[]{.22\textwidth}
        \begin{centering}

%
%
\definecolor{mycolor1}{rgb}{0.00000,0.44700,0.74100}%
\begin{tikzpicture}

\begin{axis}[%
width=0.951\figurewidth,
height=\figureheight,
at={(0\figurewidth,0\figureheight)},
scale only axis,
bar shift auto,
xmin=-0.2,
xmax=21.2,
xtick={\empty},
ymin=0,
ymax=34000,
axis background/.style={fill=white}
]
\addplot[ybar, bar width=3.1, fill=mycolor1, draw=black, area legend] table[row sep=crcr] {%
1	21586\\
2	1374\\
3	647\\
4	645\\
5	286\\
6	193\\
7	181\\
8	174\\
9	153\\
10	122\\
11	109\\
12	92\\
13	83\\
14	77\\
15	77\\
16	75\\
17	74\\
18	74\\
19	72\\
20	67\\
};
\addplot[forget plot, color=white!15!black] table[row sep=crcr] {%
-0.2	0\\
21.2	0\\
};
\end{axis}
\end{tikzpicture}%

            {\bf (b)} $\theta = 0.1$

        \end{centering}
    \end{minipage}
    \begin{minipage}[]{.22\textwidth}
        \begin{centering}

%
%
\definecolor{mycolor1}{rgb}{0.00000,0.44700,0.74100}%
\begin{tikzpicture}

\begin{axis}[%
width=0.951\figurewidth,
height=\figureheight,
at={(0\figurewidth,0\figureheight)},
scale only axis,
bar shift auto,
xmin=-0.2,
xmax=21.2,
xtick={\empty},
ymin=0,
ymax=34000,
axis background/.style={fill=white}
]
\addplot[ybar, bar width=3.1, fill=mycolor1, draw=black, area legend] table[row sep=crcr] {%
1	31258\\
2	25094\\
3	21696\\
4	16079\\
5	13482\\
6	10675\\
7	8579\\
8	4642\\
9	3452\\
10	3312\\
11	3091\\
12	2727\\
13	2675\\
14	2498\\
15	2273\\
16	2080\\
17	1983\\
18	1923\\
19	1337\\
20	1333\\
};
\addplot[forget plot, color=white!15!black] table[row sep=crcr] {%
-0.2	0\\
21.2	0\\
};
\end{axis}
\end{tikzpicture}%

            {\bf (c)} $\theta = 0.8$

        \end{centering}
    \end{minipage}
    \caption{
        Effect of replacing the stripe dictionary $\Omega$ with the convex combination
        $\Omega_\theta = (1-\theta) \Omega + \theta n^2 \bar D_l$ to
        sparse-code the image \texttt{barbara} after local contrast normalization.
        %
        {\bf (a)} The average reconstruction error
        $\| \Omega_\theta S_i \Gamma -R_i Y \|_2 $ (in blue) and the average
        Euclidean distance between patches and slices $\|n^2 \bar D_l S_i
        \Gamma - R_i Y \|_2$ (in red) as a function of $\theta$.
        In accordance to the inequality constraint, the reconstruction error remains
        below the threshold $T$.
        By construction, overlapping slices must be combined to approximate patches.
        However, as $\theta$ increases, individual slices $n^2 \bar D_l S_i
        \Gamma$ become increasingly similar to patches.
        {\bf (b)} and {\bf (c)} Number of non-zero coefficients
        for each of the 20 most commonly used atoms for $\theta = 0.1$
        and $\theta = 0.8$ respectively.
        As $\theta$ increases, more diverse local atoms are used.
    }
    \label{fig:weighted}
\end{figure}

The method described above for the $\ell_{2,\infty} - \ell_1$ formulation brings an
additional level of flexibility by offering a generic way to enforce a wider
range of structured sparsity.
Indeed, because the proposed method splits the global pursuit into
parallel pursuits on each stripe, a specific local structure can be imposed on
individual stripes.
This can be achieved naturally by simply weighting the columns of the stripe
dictionary, so as to relatively promote or penalize the use of certain atoms.
Formally this corresponds to
\begin{equation}
    \begin{aligned}
        & \minimize_{\substack{ \Gamma}}
        & &  \|\Gamma \|_1 \\
        & \text{ subject to } & &  \forall i, \;  \|\Omega W_i S_i \Gamma - R_i Y \|_2^2 \leq T, \\
    \end{aligned}
    \label{eq:csc:l2inf:reweigthed}
\end{equation}
where $W_i$ denotes the weighting diagonal matrix relative to the $i$-th
patch\footnote{Note that to be consistent with the global CSC model, the set
of matrices $\{W_i\}$ must satisfy the relation $D = \frac{1}{n^2}
\sum R_i^T \Omega W_i S_i$}.
In the context of the proposed algorithm, this boils down to an extra weighting
within each FISTA iterations.

One particularly interesting application of such strategy consists in combining the CSC
and patch-averaging models.
Such a combination allows for the benefits of both the global and local models,
which respective performances on various tasks are increasingly well
understood.
From an analysis stand point, being able to examine the entire spectrum
separating the CSC model and the patch-averaging approach is highly valuable,
as the understand of their precise inter-relation has been of interest to the
image processing community~\cite{carrera2017sparse}.
With the proposed method, such combination can be achieved via a mere
re-weighting of the columns that amounts to replacing the stripe dictionary
with the convex combination
\begin{equation}
    \Omega_\theta = (1-\theta) \Omega + \theta n^2 \bar D_l,
    \label{}
\end{equation}
with $0 \leq \theta \leq 1$ and with $\bar D_l$ denoting the local dictionary
padded with zero columns.
The parameter $\theta$ allows to regulate the level of
patch aggregation that has been proven to be critical in denoising
problems~\cite{carrera2017sparse}.
Setting $\theta=0$ corresponds to the $\ell_1 - \ell_{2,\infty}$ CSC formulation
above.
By increasing  $\theta$, filters which locations are shifted with respect to the
patch are increasingly penalized.
Setting $\theta = 1$ is synonymous with the patch averaging strategy in which the
reconstruction relies exclusively on $D_l$ and none of its shifted atoms.
The behavior of the resulting problem
\begin{equation}
    \begin{aligned}
        & \minimize_{\substack{ \Gamma}}
        & &  \|\Gamma \|_1 \\
        & \text{ subject to } & &  \forall i, \;  \|\Omega_\theta S_i \Gamma - R_i Y \|_2^2 \leq T, \\
    \end{aligned}
    \label{eq:csc:l2inf:reweigthed:pa}
\end{equation}
and the structure of its solution are examined in Figure~\ref{fig:weighted}.
Figure~\ref{fig:weighted} {\bf (a)} shows the average representation error $\|
\Omega_\theta S_i -R_i Y \|_2$ (in blue) and the average Euclidean
distance between individual slices and patches $\|n^2 \bar D_l S_i \Gamma - R_i
Y \|$  (in red) as a functions of the parameter $\theta$.
In accordance to the inequality constraints in Problem~\eqref{eq:csc:l2inf}, the
patch reconstruction error stays below the threshold $T$.
On the other hand, and as expected, the Euclidean distance between slices and
patches is above the threshold $T$, as it is the combination of
overlapping slices, rather than an isolated slice, that approximates the patch.
However, as $\theta$ increases, the term $\Omega_\theta S_i \Gamma$ in the
representation error in Problem~\eqref{eq:csc:l2inf:reweigthed:pa} is
increasingly similar to a slice $n^2 D_l \alpha $.
This in turn constrains the individual slices to better approximate the
corresponding patch.
Additionally, the constraint affects the diversity of local atoms used in the
global representation.
Indeed, Figure~\ref{fig:weighted} {\bf (b)} and {\bf (c)} show the number of non-zero
coefficients for $\theta = 0.1$ and $\theta = 0.8$ respectively.
Even though the formulations for $\theta = 0.1$ and $\theta = 0.8$ are both
consistent with the global CSC model, the latter leads to more diverse local
atoms being used.
We will see next how this behavior brings additional practical benefits.

\subsection{Experiments}

\begin{table*}[t]
    \centering
    \begin{tabular}{|l|c|c|c|c|c|c|c|}
        \hline
                          &  \texttt{barbara} & \texttt{lena} & \texttt{boat} & \texttt{hill} & \texttt{house}  & \texttt{couple} & \texttt{man}   \\
        \hline
        Heide et al. \cite{heide2015fast}              & 11.00 & 11.77 & 10.29 & 10.37 &  10.18 & 11.99  & 11.60  \\
        Papyan et al. \cite{papyan2017convolutional}   & { 11.67} & 11.92 & 10.33 & { \bf 10.66}  & { 10.56} & 12.25  & { 11.84}  \\
        $\ell_{1} - \ell_{2,\infty}$                   & 11.65  & { 11.99} & { 10.39} & 10.55 & 10.60 & { 12.34}   &  11.91 \\
        weighted $\ell_{1} - \ell_{2,\infty}$,    & {\bf 11.78}  &  {\bf 12.13} & {\bf 10.58} & { 10.65} & {\bf 10.62} & {\bf 12.46} & {\bf 11.98} \\
        \hline
        \hline
        Papyan et al. \cite{papyan2017convolutional} , image specific $D_l$  & 15.20  & {\bf 12.35} & 11.60 & 10.90 & 11.70 & 12.41 & 11.71 \\
        weighted $\ell_{1} - \ell_{2,\infty}$, image specific $D_l$   &{\bf 16.11}  & { 12.29} & {\bf 11.93} & {\bf 11.22} & {\bf 12.13} & {\bf 13.16} &{\bf  12.05} \\
        \hline
    \end{tabular}
    \caption{
        Image inpainting. The  $\ell_{2} - \ell_1$ based method
        of~\cite{papyan2017convolutional} and~\cite{heide2015fast} are compared
        to the proposed methods: the $\ell_{2,\infty} - \ell_1$ formulation and
        the formulation with a weighted stripe dictionary.
        In the first block, the local dictionary is pretrained from the
        \texttt{fruit} dataset using the method
        from~\cite{papyan2017convolutional}.
        The $\ell_{2,\infty}$ prior improves over the best $\ell_2 - \ell_1$
        based method formulation.
        The weighted stripe dictionary $\Omega_\theta$ with $\theta = 0.8$
        brings an additional improvement in PSNR over the standard
        $\ell_{2,\infty}$ by promoting patch averaging.
        In the result reported in the second block, the local dictionary used
        is learned from the corrupted image.
        In this scenario, the weighted $\ell_{2,\infty} - \ell_1$ formulation
        with $\theta = 0.8$ generally
        outperforms~\cite{papyan2017convolutional}.
    }
    \label{table:inpainting:psnr}
\end{table*}

We illustrate the behavior of the $\ell_{2,\infty} - \ell_1$ formulation and its weighted variant on the
classic problem of image inpainting.
Let us consider an image $X$ and a diagonal binary matrix $M$, which masks the
entries in $X$ in which $M_{i,i} = 0$.
Image inpainting is the process of filling in missing areas in an image in a
realistic manner.
That is, given the corrupted image $Y=MX$, the task consists in estimating the
original signal $X$.

Estimating the original signal via the $\ell_{2,\infty} - \ell_{1}$ CSC
requires solving the problem
\begin{equation}
    \begin{aligned}
        & \minimize_{\substack{ \Gamma}}
        & &  \|\Gamma \|_1 \\
        & \text{ subject to } & &  \forall i, \;  \| R_i ( M D \Gamma -  Y ) \|_2^2 \leq T_i, \\
    \end{aligned}
    \label{eq:csc:l2inf:inpainting}
\end{equation}
where the constraint on the representation accuracy incorporates the binary
matrix $M,$ and where the threshold $T_i$ is set on a patch-by-patch
basis to reflect the varying numbers of active pixels in each patch.
Minimizing this objective requires only a slight modification of the algorithm
described above, namely incorporating the mask into the function $g_i$ and its
gradient.
The PPXA relaxation parameter is set to $\lambda_l = 1.6$ and the scaling
factor in the proximal operator is set to $\mu = 100$.

Table~\ref{table:inpainting:psnr} contains the peak signal-to-noise ratio
(PSNR) on a set of publicly available standard test images.
In the first block of experiments, we adopt the benchmark framework proposed
in~\cite{heide2015fast}.
In particular, the local contrast normalization is applied to the input image
and the local dictionary is pretrained from the \texttt{fruit} dataset, using
the method from~\cite{papyan2017convolutional}.
The method based on the $\ell_{2,\infty} - \ell_{1}$ formulation outperforms
the method proposed in~\cite{heide2015fast} and slightly improves over the
slice-based approach of~\cite{papyan2017convolutional}.
The best performance are obtained in general with the weighted $\ell_{2,\infty}
- \ell_1$ ($\theta=0.8$), which formulation tends to promote an averaging of
similar local estimates.
Significant additional improvements are achieved when learning the local
dictionary $D_l$ from the corrupted image.
The second block in Table~\ref{table:inpainting:psnr} contains the inpainting
PSNR obtained in this scenario for the sliced based
method~\cite{papyan2017convolutional} and for the weighted $\ell_{2,\infty}-
\ell_1$ used along the dictionary update proposed in~\cite{plaut2018mpcsc}.
In this context, the weighting of the stripe dictionary is particularly
beneficial as it encourages more atoms to be used and therefore updated  (see
Figure~\ref{fig:weighted}).

\section{Conclusion}
\label{sec:conclusion}

While enjoying a renewed interest in recent years, the CSC model has been
almost exclusively considered in its $\ell_2 - \ell_1$ formulation.
In the present work, we expanded the formulations for the CSC  with two alternative
formulations, namely the $\ell_2 - \ell_{1,\infty}$  and $\ell_{2,\infty} -
\ell_{1}$  formulations in which mixed-norms, alter how the spatial
distributions of non-zero coefficients are controlled.
For both formulations, we derived algorithms that rely on the ADMM and PPXA
algorithms.
The algorithms are simple, easy to implement and can take full advantage of
fast GPU implementation of the convolution operator.
Their convergence naturally follows from the convergence properties of the two
standard convex optimization framework they build on.
We examined the performance and behavior of the proposed formulation on two
image processing tasks: inpainting and cartoon texture separation.
Furthermore, we showed that the  $\ell_{2,\infty} - \ell_1$ formulation in
particular opens the door to a wide variety of structured sparsity, that could
bring additional practical benefits while still being consistent with the CSC
model.
An interesting example of such structured sparsity was offered in the
combination of the CSC and patch-averaging models, showing that such a mixture
provides improved performance.
Finally, we envision that similar combinations of global and local sparse
priors, within the proposed unifying framework, will allow to further benefits
in several other restoration problems.

\section{Acknowledgements}

The research leading to these results has received funding in part from the
European Research Council under EU’s 7th Framework Program, ERC under Grant
320649, and in part by Israel Science Foundation (ISF) grant no. 1770/14.

\newpage

\bibliographystyle{IEEEbib}
\bibliography{main_ppxa}

\begin{thebibliography}{10}

\bibitem{bookSparse}
Michael Elad,
\newblock {\em Sparse and Redundant Representations - From Theory to
  Applications in Signal and Image Processing.},
\newblock Springer, 2010.

\bibitem{elad2006image}
Michael Elad and Michal Aharon,
\newblock ``Image denoising via sparse and redundant representations over
  learned dictionaries,''
\newblock {\em IEEE Transactions on Image processing}, vol. 15, no. 12, pp.
  3736--3745, 2006.

\bibitem{mairal2009non}
Julien Mairal, Francis Bach, Jean Ponce, Guillermo Sapiro, and Andrew
  Zisserman,
\newblock ``Non-local sparse models for image restoration,''
\newblock in {\em Computer Vision, 2009 IEEE 12th International Conference on}.
  IEEE, 2009, pp. 2272--2279.

\bibitem{romano2014single}
Yaniv Romano, Matan Protter, and Michael Elad,
\newblock ``Single image interpolation via adaptive non-local sparsity-based
  modeling,''
\newblock {\em IEEE Transactions on Image Processing}, 2014.

\bibitem{yang2010image}
Jianchao Yang, John Wright, Thomas~S Huang, and Yi~Ma,
\newblock ``Image super-resolution via sparse representation,''
\newblock {\em IEEE transactions on image processing}, vol. 19, no. 11, pp.
  2861--2873, 2010.

\bibitem{yu2012solving}
Guoshen Yu, Guillermo Sapiro, and St{\'e}phane Mallat,
\newblock ``Solving inverse problems with piecewise linear estimators: From
  gaussian mixture models to structured sparsity,''
\newblock {\em IEEE Transactions on Image Processing}, vol. 21, no. 5, pp.
  2481--2499, 2012.

\bibitem{dong2013nonlocally}
Weisheng Dong, Lei Zhang, Guangming Shi, and Xin Li,
\newblock ``Nonlocally centralized sparse representation for image
  restoration,''
\newblock {\em IEEE Transactions on Image Processing}, vol. 22, no. 4, pp.
  1620--1630, 2013.

\bibitem{aharon2006rm}
Michal Aharon, Michael Elad, and Alfred Bruckstein,
\newblock ``K-svd: An algorithm for designing overcomplete dictionaries for
  sparse representation,''
\newblock {\em IEEE Transactions on signal processing}, vol. 54, no. 11, pp.
  4311--4322, 2006.

\bibitem{engan1999method}
Kjersti Engan, Sven~Ole Aase, and J~Hakon Husoy,
\newblock ``Method of optimal directions for frame design,''
\newblock in {\em Acoustics, Speech, and Signal Processing, 1999. Proceedings.,
  1999 IEEE International Conference on}. IEEE, 1999, vol.~5, pp. 2443--2446.

\bibitem{mairal2008sparse}
Julien Mairal, Michael Elad, and Guillermo Sapiro,
\newblock ``Sparse representation for color image restoration,''
\newblock {\em IEEE Transactions on image processing}, vol. 17, no. 1, pp.
  53--69, 2008.

\bibitem{sulam2015expected}
Jeremias Sulam and Michael Elad,
\newblock ``Expected patch log likelihood with a sparse prior,''
\newblock in {\em International Workshop on Energy Minimization Methods in
  Computer Vision and Pattern Recognition}. Springer, 2015, pp. 99--111.

\bibitem{papyan2016multi}
Vardan Papyan and Michael Elad,
\newblock ``Multi-scale patch-based image restoration,''
\newblock {\em IEEE Transactions on image processing}, vol. 25, no. 1, pp.
  249--261, 2016.

\bibitem{mairal2008learning}
Julien Mairal, Guillermo Sapiro, and Michael Elad,
\newblock ``Learning multiscale sparse representations for image and video
  restoration,''
\newblock {\em Multiscale Modeling \& Simulation}, vol. 7, no. 1, pp. 214--241,
  2008.

\bibitem{sulam2014image}
Jeremias Sulam, Boaz Ophir, and Michael Elad,
\newblock ``Image denoising through multi-scale learnt dictionaries,''
\newblock in {\em Image Processing (ICIP), 2014 IEEE International Conference
  on}. IEEE, 2014, pp. 808--812.

\bibitem{zoran2011learning}
Daniel Zoran and Yair Weiss,
\newblock ``From learning models of natural image patches to whole image
  restoration,''
\newblock in {\em Computer Vision (ICCV), 2011 IEEE International Conference
  on}. IEEE, 2011, pp. 479--486.

\bibitem{grosse2012shift}
Roger Grosse, Rajat Raina, Helen Kwong, and Andrew~Y Ng,
\newblock ``Shift-invariance sparse coding for audio classification,''
\newblock {\em arXiv preprint arXiv:1206.5241}, 2012.

\bibitem{thiagarajan2008shift}
Jayaraman Thiagarajan, Karthikeyan Ramamurthy, and Andreas Spanias,
\newblock ``Shift-invariant sparse representation of images using learned
  dictionaries,''
\newblock in {\em Machine Learning for Signal Processing, 2008. MLSP 2008. IEEE
  Workshop on}. IEEE, 2008, pp. 145--150.

\bibitem{rusu2014explicit}
Cristian Rusu, Bogdan Dumitrescu, and Sotirios~A Tsaftaris,
\newblock ``Explicit shift-invariant dictionary learning,''
\newblock {\em IEEE Signal Processing Letters}, vol. 21, no. 1, pp. 6--9, 2014.

\bibitem{bristow2013fast}
Hilton Bristow, Anders Eriksson, and Simon Lucey,
\newblock ``Fast convolutional sparse coding,''
\newblock in {\em Proceedings of the IEEE Conference on Computer Vision and
  Pattern Recognition}, 2013, pp. 391--398.

\bibitem{heide2015fast}
Felix Heide, Wolfgang Heidrich, and Gordon Wetzstein,
\newblock ``Fast and flexible convolutional sparse coding,''
\newblock in {\em Proceedings of the IEEE Conference on Computer Vision and
  Pattern Recognition}, 2015, pp. 5135--5143.

\bibitem{kong2014fast}
Bailey Kong and Charless~C. Fowlkes,
\newblock ``Fast convolutional sparse coding (fcsc),''
\newblock {\em Department of Computer Science, University of California,
  Irvine, Tech. Rep}, vol. 3, 2014.

\bibitem{wohlberg2014efficient}
Brendt Wohlberg,
\newblock ``Efficient convolutional sparse coding,''
\newblock in {\em Acoustics, Speech and Signal Processing (ICASSP), 2014 IEEE
  International Conference on}. IEEE, 2014, pp. 7173--7177.

\bibitem{gu2015convolutional}
Shuhang Gu, Wangmeng Zuo, Qi~Xie, Deyu Meng, Xiangchu Feng, and Lei Zhang,
\newblock ``Convolutional sparse coding for image super-resolution,''
\newblock in {\em Proceedings of the IEEE International Conference on Computer
  Vision}, 2015, pp. 1823--1831.

\bibitem{yellin2017blood}
Florence Yellin, Benjamin~D. Haeffele, and Ren{\'e} Vidal,
\newblock ``Blood cell detection and counting in holographic lens-free imaging
  by convolutional sparse dictionary learning and coding,''
\newblock in {\em Biomedical Imaging (ISBI 2017), 2017 IEEE 14th International
  Symposium on}. IEEE, 2017, pp. 650--653.

\bibitem{serrano2016convolutional}
Ana Serrano, Felix Heide, Diego Gutierrez, Gordon Wetzstein, and Belen Masia,
\newblock ``Convolutional sparse coding for high dynamic range imaging,''
\newblock in {\em Computer Graphics Forum}. Wiley Online Library, 2016,
  vol.~35, pp. 153--163.

\bibitem{papyan2017working}
Vardan Papyan, Jeremias Sulam, and Michael Elad,
\newblock ``Working locally thinking globally: Theoretical guarantees for
  convolutional sparse coding,''
\newblock {\em IEEE Transactions on Signal Processing}, vol. 65, no. 21, pp.
  5687--5701, 2017.

\bibitem{plaut2018mpcsc}
Elad Plaut and Raja Giryes,
\newblock ``Matching pursuit based convolutional sparse coding,''
\newblock in {\em Acoustics, Speech and Signal Processing (ICASSP), 2018 IEEE
  International Conference on}. IEEE, 2018, IEEE SigPort.

\bibitem{wohlberg2017convolutional}
Brendt Wohlberg,
\newblock ``Convolutional sparse coding with overlapping group norms,''
\newblock {\em arXiv preprint arXiv:1708.09038}, 2017.

\bibitem{papyan2017convolutional}
Vardan Papyan, Yaniv Romano, Michael Elad, and Jeremias Sulam,
\newblock ``Convolutional dictionary learning via local processing.,''
\newblock in {\em ICCV}, 2017, pp. 5306--5314.

\bibitem{kowalski2009sparse}
Matthieu Kowalski,
\newblock ``Sparse regression using mixed norms,''
\newblock {\em Applied and Computational Harmonic Analysis}, vol. 27, no. 3,
  pp. 303--324, 2009.

\bibitem{boyd2011distributed}
Stephen Boyd, Neal Parikh, Eric Chu, Borja Peleato, Jonathan Eckstein, et~al.,
\newblock ``Distributed optimization and statistical learning via the
  alternating direction method of multipliers,''
\newblock {\em Foundations and Trends{\textregistered} in Machine learning},
  vol. 3, no. 1, pp. 1--122, 2011.

\bibitem{kelley1999iterative}
Carl~T Kelley,
\newblock {\em Iterative methods for optimization}, vol.~18,
\newblock Siam, 1999.

\bibitem{duchi2008efficient}
John Duchi, Shai Shalev-Shwartz, Yoram Singer, and Tushar Chandra,
\newblock ``Efficient projections onto the l 1-ball for learning in high
  dimensions,''
\newblock in {\em Proceedings of the 25th international conference on Machine
  learning}. ACM, 2008, pp. 272--279.

\bibitem{chan2011augmented}
Stanley~H Chan, Ramsin Khoshabeh, Kristofor~B Gibson, Philip~E Gill, and
  Truong~Q Nguyen,
\newblock ``An augmented lagrangian method for total variation video
  restoration,''
\newblock {\em IEEE Transactions on Image Processing}, vol. 20, no. 11, pp.
  3097--3111, 2011.

\bibitem{beck2009fast}
Amir Beck and Marc Teboulle,
\newblock ``A fast iterative shrinkage-thresholding algorithm for linear
  inverse problems,''
\newblock {\em SIAM journal on imaging sciences}, vol. 2, no. 1, pp. 183--202,
  2009.

\bibitem{combettes2008proximal}
Patrick~L Combettes and Jean-Christophe Pesquet,
\newblock ``A proximal decomposition method for solving convex variational
  inverse problems,''
\newblock {\em Inverse problems}, vol. 24, no. 6, pp. 065014, 2008.

\bibitem{combettes2011proximal}
Patrick~L Combettes and Jean-Christophe Pesquet,
\newblock ``Proximal splitting methods in signal processing,''
\newblock in {\em Fixed-point algorithms for inverse problems in science and
  engineering}, pp. 185--212. Springer, 2011.

\bibitem{parikh2014proximal}
Neal Parikh, Stephen Boyd, et~al.,
\newblock ``Proximal algorithms,''
\newblock {\em Foundations and Trends{\textregistered} in Optimization}, vol.
  1, no. 3, pp. 127--239, 2014.

\bibitem{bauschke2011convex}
Heinz~H Bauschke, Patrick~L Combettes, et~al.,
\newblock {\em Convex analysis and monotone operator theory in Hilbert spaces},
  vol. 408,
\newblock Springer, 2011.

\bibitem{daubechies2004iterative}
Ingrid Daubechies, Michel Defrise, and Christine De~Mol,
\newblock ``An iterative thresholding algorithm for linear inverse problems
  with a sparsity constraint,''
\newblock {\em Communications on Pure and Applied Mathematics: A Journal Issued
  by the Courant Institute of Mathematical Sciences}, vol. 57, no. 11, pp.
  1413--1457, 2004.

\bibitem{carrera2017sparse}
Diego Carrera, Giacomo Boracchi, Alessandro Foi, and Brendt Wohlberg,
\newblock ``Sparse overcomplete denoising: aggregation versus global
  optimization,''
\newblock {\em IEEE Signal Processing Letters}, vol. 24, no. 10, pp.
  1468--1472, 2017.

\end{thebibliography}

\appendix

\section{Appendix}\label{sec:appendix}

\newtheorem{myprop}{Proposition}

\begin{myprop}
    For a point $y$ and the $\ell_1$-ball of radius $r$, $\mathcal{B}_r :=\{ x,
    \text{s.t.} \|x\|_1 \leq r \}$, the distance between $y$ and the ball
        \begin{equation*}
            d(y, \mathcal{B}_r) := \inf \left\{ \|x - y \|_2, \; | \; x \in \mathcal{B}_r  \right\},
        \end{equation*}
    is a convex function of the ball radius $r$.
    \label{prop:conv}
\end{myprop}

\begin{proof}

From the $\ell_1$-norm triangle inequality, it comes that for any
convex combination of two radii $ \theta r_1 + (1-\theta) r_2$, with $0 \leq
\theta \leq 1$, we have the inclusion
$$ \theta \mathcal{B}_{r_1} + (1-\theta) \mathcal{B}_{r_2}  \subset
\mathcal{B}_{\theta r_1 + (1-\theta) r_2},$$
where $ \theta \mathcal{B}_{r_1}$ denotes the set of points
$ \{  \theta x_1  | x_1 \in \mathcal{B}_{r_1} \}.$
In particular, for the nearest points to $y$ in $\mathcal{B}_{r_1}$ and $\mathcal{B}_{r_2}$ respectively, \ie for $x_1
\in \mathcal{B}_{r_1}$ such that $\|y-x_1\|_2 = d(y, \mathcal{B}_{r_1})$
and $x_2 \in \mathcal{B}_{r_2}$ such that $\|y-x_2\|_2 = d(y,
\mathcal{B}_{r_2})$, we have
$$\theta x_1 + (1-\theta) x_2 \in \mathcal{B}_{\theta r_1 + (1-\theta) r_2},$$
and therefore
$$\|y - ( \theta x_1 + (1-\theta) x_2) \|_2 \geq
d(y, \mathcal{B}_{\theta r_1 + (1-\theta) r_2}).$$
Finally, from the Euclidean norm  triangle inequality, it comes
that
$$\theta d(y,\mathcal{B}_{r_1}) + (1-\theta) d(y, \mathcal{B}_{r_2})
\geq  d(y, \mathcal{B}_{\theta r_1 + (1-\theta) r_2})$$
which proves that $r
\mapsto d(y, \mathcal{B}_r )$ is convex.

\end{proof}

\end{document}